\documentclass{article}

\usepackage[utf8]{inputenc}
\usepackage{enumerate}
\usepackage{amsfonts}
\usepackage{amsmath}
\usepackage{amsthm}
\usepackage{blindtext}
\usepackage{graphicx}
\usepackage[numbers]{natbib}
\usepackage{amssymb}
\usepackage{mathtools}
\usepackage{stmaryrd}
\usepackage{tikz-cd}
\usepackage{relsize}
\usepackage{mathrsfs}
\usepackage{ upgreek }
\usepackage[normalem]{ulem}
\usepackage{qcircuit}
\usepackage{braket}
\usepackage{authblk}
\usepackage{blkarray}

\newtheorem{theorem}{Theorem}
\newtheorem{lemma}{Lemma}

\newtheorem{proposition}{Proposition}
\theoremstyle{definition}

\newtheorem{proposition-definition}{Proposition-Definition}

\graphicspath{ {./images/} }
\numberwithin{figure}{section}

\title{Clifford circuits over non-cyclic abelian groups}

\author[1,2]{Milo Moses}
\author[2,3]{Jacek Horecki}
\author[2,3]{Konrad Deka}
\author[2,4]{Jan~Tułowiecki}
\affil[1]{Department of Mathematics, UC Santa Barbara}
\affil[2]{BEIT Inc.}
\affil[3]{Institute of Mathematics, Jagiellonian University}
\affil[4]{Department of Theoretical Computer Science, Jagiellonian University}

\begin{document}

\maketitle

\newcommand{\RR}{\mathbb{R}}
\newcommand{\HH}{\mathbb{H}}
\newcommand{\NN}{\mathbb{N}}
\newcommand{\QQ}{\mathbb{Q}}
\newcommand{\CC}{\mathbb{C}}
\newcommand{\FF}{\mathcal{F}}
\newcommand{\ZZ}{\mathbb{Z}}
\newcommand{\st}{\,\,\mathrm{s.t.}\,\,}
\newcommand{\mm}{\mathfrak{m}}
\newcommand{\pp}{\mathfrak{p}}
\newcommand{\Hom}{\mathrm{Hom}}
\newcommand{\Aut}{\mathrm{Aut}}
\newcommand{\Cat}{\bold{Cat}}
\newcommand{\nullclass}{\left|\bold{0}\right>}
\newcommand{\alphaclass}{\left|\alpha\right>}
\newcommand{\betaclass}{\left|\beta\right>}
\newcommand{\alphabetaclass}{\left|\alpha\beta\right>}
\newcommand{\ppsi}{\left|\psi\right>}
\newcommand{\pphi}{\left|\phi\right>}
\newcommand{\vin}{\rotatebox[origin=c]{-90}{$\in$}}
\newcommand{\Ghat}{\widehat{G}}
\newcommand{\Pauli}{\mathrm{Pauli}}
\newcommand{\Cliff}{\mathrm{Cliff}}
\newcommand{\CX}{\text{C}X}
\newcommand{\polylog}{\mathrm{polylog}}
\newcommand{\diag}{\mathrm{diag}}
\newcommand{\Sp}{\mathrm{Sp}}

\begin{abstract}
We present a discussion of the generalized Clifford group over non-cyclic finite abelian groups. These Clifford groups appear naturally in the theory of topological error correction and abelian anyon models. We demonstrate a generalized Gottesman-Knill theorem, stating that every Clifford circuit can be efficiently classically simulated. We additionally provide circuits for a universal quantum computing scheme based on local two-qudit Clifford gates and magic states.
\end{abstract}

\section{Introduction}

The creation of a large-scale quantum computer is the ambitious goal of many scientists. One of the major steps towards the realization of this goal is the creation of low-overhead error correcting codes on which a universal set of logical gates can be applied using fault-tolerant methods. These error correcting codes are typically stabilizer codes, and the fault-tolerant gates are typically made from constant-depth local unitary circuits.

It is here that one runs into the fundamental issue of error correction: constant-depth local unitary circuits on two dimensional stabilizer codes can only act on the codespace by Clifford gates. This is known as the  Bravyi-Kong bound \cite{bravyi2013classification}, and along with related no-go theorems like the Eastin-Knill theorem \cite{eastin2009restrictions} it sets out a clear vision of where the difficulty in error correction lies. Not only is the Clifford gate a proper subset of all possible quantum gates, but the celebrated Gottesman-Knill theorem \cite{gottesman1998heisenberg} asserts that the result of any quantum computation using only Clifford gates can be efficiently simulated on a classical computer. Hence, the future of quantum computation relies on the use of tricks to circument no-go theorems. For this reason as well as others, Clifford gates have attracted enormous attention \cite{bravyi2005universal, hostens2005stabilizer, brown2017poking, zeng2008semi, rengaswamy2019unifying, cui2017diagonal}.

In this paper we will be deeply analyzing computations with a generalized version of the Clifford group. To understand this generalization, it is best to recall the history of topological quantum computation. The first topological codes have their origins in an unexpected connection to condensed matter physics \cite{kitaev2003fault}. They were discovered by Kitaev in his study of discrete gauge theories, where it was shown that the mathematics governing topological phases of matter directly yield error correction schemes \cite{kitaev1997quantum}. When the gauge group is $G=\ZZ_2$, this is the surface code. In general, the gauge group $G$ can be an arbitrary finite group.

When $G$ is non-abelian, the theory quickly becomes extremely subtle \cite{bombin2008family, yan2022ribbon, cowtan2022quantum}. When $G$ is abelian, however, the error correction procedure is unitary and efficient to perform \cite{anwar2014fast, bullock2007qudit}. The associated stabilizer formalism is not based on the Pauli group, but is instead based on a generalized ``$G$-Pauli group".

Interestingly enough, these $G$-Pauli groups were already known to quanutm information reserachers a few years before Kitaev's introduction of topological quantum computing \cite{knill1996group, knill1996non}. Overwhelmingly, however, when quantum information researchers study generalized stabilizer formalisms they focus on the case where $G=\ZZ_d$ is cyclic \cite{hostens2005stabilizer, de2013linearized, duclos2013kitaev, gottesman1998fault, wang2020qudits}. This focus on $G=\ZZ_d$ also appears in condensed matter physics circles, because cyclic groups arise from symmetry breaking of $U(1)$ gauge theories whereas non-cyclic abelian groups must arrise from $U(1)^N$ Chern-Simons thoeries with $N\geq 2$ \cite{cano2014bulk}.

The case of general abelian groups $G$ has been neglected to such an extent that some statements about abelian topological error correcting codes which are accepted as fact in topological quantum computing commuties have yet to explicitly appear in literature. This article demonstrates some of these results. Specifically, we prove a generalized Gottesman-Knill theorem showing that $G$-Clifford gates can be efficiently classically simulated, as well as writing circuits for a``Clifford + magic state"-type universal quantum computation scheme.

There are many different topoloigcal quantum error correcting schemes, all of which require slightly different logical state encodings and gate implimentations \cite{brown2017poking}. For this reason we restrict our attention to an implementation-agnostic description of the circuits, which only relies on the general picture. This general picture states that it is straightforward to construct one-qudit Pauli operations. Using some extra techniques, it is not too difficult to construct one-qudit Clifford operations. Using some sort of procedure such as merging lattice squares or braiding holes, some two-qudit Clifford gates can be implemented. Past that, there are no other native topologically protected operations. Using a magic state distillation technique one can obtain qubits prepared in predetermined states to high fidelity \cite{bravyi2005universal}. Quantum teleportation is then used to perform universal quantum computation \cite{gottesman1999quantum}.

One subtelty which does not appear in the case where $G=\ZZ_d$ is cyclic is that the $G$-Clifford group is $\textit{not}$ generated by one-qudit Clifford gates and the two-qudit $\CX$ gate. It is true that every Clifford gate can be generated by one and two-qudit gates, but the $\CX$ gate is not enough to generate all two-qudit gates. It is not clear how one would implement these exotic gates using lattice surgery, so this is a topic for future research.

Some of the most relevant works to this current discussion are a series of papers studying the power of quantum computation based on the $G$-Clifford group for finite abelian groups $G$ \cite{bermejo2014classical, van2013efficient}. These papers prove a Gottesman-Knill-type theorem for $G$-Clifford circuits showing that they can be efficiently classically simulated in polynomial time, generalizing results for qubits and qudits \cite{gottesman1998heisenberg, de2013linearized}. This Gottesman-Knill theorem is proved with a caveat, however. It only applies to circuits in a distinguished subclass of Clifford gates generated by Fourier transforms, automorphism gates, and quadratic phase gates. An outstanding conjecture is to demonstrate that these gates generate all Clifford gates, and thus that the generalized Gottesman-Knill theorem holds without extra assumptions. The proof of this conjecture is a straightforward group-theoretic generatlization of the qudit linear algebra performed in \cite{hostens2005stabilizer}, which we give in Section \ref{generating set appendix}. We additionally show that this gate set has redundancies which can be removed. These older works make no explicit connection between their results and error correcting codes, and they seem to have gone virtually unnoticed by the topological quantum computing community.

Even though $G$-Clifford circuits can be efficiently classically simulated for all finite abelian groups $G$, there still can be benefits of computing on groups larger than $\ZZ_2$. When $G=\ZZ_3$, projective charge measurements on boundary defects allow for universal quantum computation, which is epsically relevant becasue the $G=\ZZ_3$ quantum double is predicted to describe the $\nu=1/3$ filling of the fractional quantum Hall system \cite{cong2017universal, cong2016topological}. At the ends of defect lines in abelian models one gets parafermions. When $G=\ZZ_2$ there parafermions are Majorana fermions, which do not natively allow for universal quantum computation. When $d\geq 3$ however these parafermions can be more powerful. This leads to an appealing universal quantum computing scheme based on $G=\ZZ_4$. In this case the quantum double can be nicely placed on a Kagome lattice with two-body nearest neighbor interactions, making it amenable to quantum processors \cite{hutter2015parafermions, hutter2016quantum}. Going further, the $G=\ZZ_6$ quantum double model has found applications in schemes for simulating powerful non-abelian quantum memories \cite{wootton2011engineering, wootton2010dissecting}. We hope that the results of this paper will help us find applications of quantum computing with non-cyclic abelian finite gauge groups.

In Sections \ref{pauli-group} and \ref{clifford-group} we introduce the $G$-Pauli and $G$-Clifford groups and demonstrate their key properties, some of which are new to the literature.

In Section \ref{CX gate} we give an algorithm to implement the $\CX$ gate using 2-qudit $G$-Pauli measurements.

In Section \ref{magic-states} we give an algorithm to implement non-Clifford gates using magic state injection.

In Section \ref{fourier-transform} we demonstrate that the global Fourier transform on any finite abelian group $G$ can be used to obtain Fourier transforms on split subgroups of $G$. This greatly simplifies our generating set for the Clifford gates.

In Section \ref{generating set appendix} we show that the Clifford group is generated by the subset of Fourier, automorphism, and controlled phase gates.

In Section \ref{conclusion} we give a conclusion, and list some problems which remain open in this area.

\section{The Pauli Group}
\label{pauli-group}

The Pauli matrices are core to quantum computing, and to quantum physics as a whole. They serve as our most elementary qubit operations, and are defined as

$$I=
\begin{pmatrix}
1 & 0\\
0 & 1
\end{pmatrix},
X=
\begin{pmatrix}
0 & 1\\
1 & 0
\end{pmatrix},
Y=
\begin{pmatrix}
0 & -i\\
i & 0
\end{pmatrix},
Z=
\begin{pmatrix}
1 & 0\\
0 & -1
\end{pmatrix}.$$

The special properties of these Pauli operators ies intimately linked to the algebraic structure of the group $\ZZ_2$. In fact, they can be readily generalized to an arbitrary finite abelian groups $G$ as follows. Instead of acting on the qubit $\CC^2$, the Pauli operators on $G$ (also known as ``$G$-Pauli operators") act on the group algebra $\CC[G]=\text{span}\{\left|g\right>,\,\, g\in G\}$ as follows:

\begin{align*}
X_g: \CC[G]&\xrightarrow{}\CC[G],\,\, g\in G\\
\left| h\right> &\mapsto \left| gh\right>
\end{align*}

\begin{align*}
Z_\chi: \CC[G]&\xrightarrow{}\CC[G],\,\, \chi\in \Ghat\\
\left| g\right> &\mapsto \chi(g)\left| g\right>
\end{align*}

where $\Ghat=\Hom(G,\CC^\times)$ is the character group of $G$, written in multiplicative notation. The full space of Pauli operators is the group generated by the $X_g,Z_\chi$ up to global phases in $U(1)=\{\lambda\in \CC \st |\lambda|=1\}$,

$$\Pauli(G)=\left\{\omega X_g Z_\chi,\,\, (g,\chi)\in G\times \Ghat,\,\, \omega\in U(1)\right\}.$$

We summarize the basic properties of these operators below:

\begin{proposition} Let $G$ be a finite abelian group. The following statements about $\Pauli(G)$ are all true:

\begin{enumerate}
\item All Pauli operators are unitary.

\item We have that $X_{g_0}X_{g_1}=X_{g_0g_1}$ and $Z_{\chi_0}Z_{\chi_1}=Z_{\chi_0\chi_1}$ for all $g_0,g_1\in G$, $\chi_0,\chi_1\in \Ghat$. In particular $X_g$ operators commute with each other, as do $Z_\chi$ operators .

\item $Z_\chi X_g=\chi(g)X_gZ_\chi$ for all $(g,\chi)\in G\times \Ghat$.

\item If a unitary operator $U:\CC[G]\to\CC[G]$ commutes with all elements of $\Pauli(G)$ it must be the identity.

\item If a unitary operator $U:\CC[G]\to\CC[G]$ commutes with all elements of $\Pauli(G)$ up to phase it must be a Pauli operator.

\end{enumerate}
\end{proposition}
\begin{proof} Points (1)-(3) follow immediately from the definitions of $X_g$ and $Z_\chi$. For point (4), suppose that 

$$U \ket{1}=\sum_{g\in G} c_g\ket{g}.$$

We have that

$$Z_\chi U\ket{1}=UZ_\chi\ket{1}=U\ket{1}\,\, \forall \chi\in\Ghat$$

and hence $\chi(g)c_g=c_g$ for all $\chi\in\Ghat$. Thus, $U\ket{1}=\lambda \ket{1}$ for some $\lambda\in\CC$. Now,

$$U\ket{g}=UX_g\ket{1}=X_gU\ket{1}=\lambda \cdot X_g\ket{1}=\lambda \ket{g}$$

and hence $U=\lambda I$. Thus, $U$ is the identity up to a global phase. We now move on to point (5). By assumption, $Z_\chi U$ equals  $UZ_\chi$ up to a phase for every $\chi\in \Ghat$. The map $\Ghat\to\CC^\times$ sending $\chi$ to this phase is clearly a homomorphism. Thus, since $G$ is isomorphic to its double dual, there exists $g\in G$ such that $Z_\chi U=\chi(g)UZ_\chi$ for all $\chi\in\Ghat$. If $U \ket{1}=\sum_{h\in G} c_h\ket{h}$ as before, then we have $\chi(h)c_h=\chi(g)c_h$. Thus, $U\ket{1}=\lambda \ket{g}$.

Now, we again know by assumption that $UX_g$ equals $X_g U$ up to a phase. The map $G\to\CC^\times$ sending a group element to the phase it induces is again a homomorphism, and thus there exists a character $\chi$ such that $UX_g=\chi(g)X_gU$. Hence,

$$U\ket{h}=UX_h\ket{1}=\chi(h)X_hU\ket{1}=\lambda \cdot \chi(h)X_h\ket{g}=\lambda\cdot \chi(h)\ket{hg}.$$

Thus, $U=\lambda \cdot X_g Z_\chi$ so $U$ is a Pauli operator as desired.
\end{proof}

The $G$-Pauli group on $n$ qudits is the space of operators on $\CC[G]^n$ which are the tensor products of $G$-Pauli matrices on each component. Alternatively, one can identify this group with $\Pauli\left(G^n\right)$ via the canonical isomorphism $\CC[G^n]\cong \CC[G]^n$. Note however that we do \textit{not} have an isomorphism $\Pauli(G\times H)\cong \Pauli(G)\otimes \Pauli(H)$. This is because if $\lambda\in U(1)$ is a phase, $\lambda\otimes 1$ and $1\otimes \lambda$ are different objects in $\Pauli(G)\otimes \Pauli(H)$. Thus, to make these two objects equal one has to quotent by the relation $1\otimes \lambda=\lambda\otimes1$. Or, more abstractly, we have that $\Pauli(G\times H)=\Pauli(G)\times_{U(1)}\Pauli(H)$ where the product is taken in the coslice-category of abelian groups with distinguished $U(1)$ subgroups.

\section{The Clifford group}
\label{clifford-group}

The set of gates which can be fault tolerantly implemented in the $G$ surface code is the $G$ Clifford group. This group can be succinctly described as the normalizer of the Pauli group. That is,

$$\Cliff(G)=\{U:\CC[G]\to\CC[G] \st UM U^\dagger\in \Pauli(G),\,\, \forall M\in \Pauli(G)\}.$$

It helps to work instead with an explicit generating set for $\Cliff(G)$. It is a standard fact that $\Cliff\left(\ZZ_2^n\right)$ is generated by the gates $H$, $S$, and $\CX$ where

\begin{align*}
H:\CC[\ZZ_2]&\xrightarrow{}\CC[\ZZ_2]\\
\left|0\right>&\mapsto \frac{\left|0\right>+\left|1\right>}{\sqrt{2}}\\
\left|1\right>&\mapsto \frac{\left|0\right>-\left|1\right>}{\sqrt{2}}
\end{align*}

is the Hadamard matrix,

\begin{align*}
S:\CC[\ZZ_2]&\xrightarrow{}\CC[\ZZ_2]\\
\left|0\right>&\mapsto \left|0\right>\\
\left|1\right>&\mapsto e^{\pi i /2}\left|1\right>
\end{align*}

is the $\pi/2$ phase gate, and

\begin{align*}
\CX: \CC[\ZZ_2\times \ZZ_2]&\xrightarrow{}\CC[\ZZ_2\times \ZZ_2].\\
\left|00\right>\mapsto \left|00\right>,&\,\, \left|01\right>\mapsto \left|01\right>\\
\left|10\right>\mapsto \left|11\right>,&\,\, \left|11\right>\mapsto \left|10\right>\\
\end{align*}

is the two-qubit controlled $X$ gate. These are generalized to the group theoretical setting as follows:

$$\text{Hadamard }\implies\text{Fourier transform}$$

$$\pi/2\text{ phase}\implies\text{Quadratic phase gate}$$

$$\CX\implies \text{Automorphism gates}$$

In Theorem \ref{main-theorem-1} we will show that these three families are enough to generate all Clifford gates. The automorphism gates $A_\tau$ are parameterized by automorphisms $\tau:G\to G$. They act by

\begin{align*}
A_\tau: \CC[G]&\to\CC[G].\\
\left| g\right>&\mapsto \left|\tau(g)\right>
\end{align*}

For example, if $G=\ZZ_2\times \ZZ_2$ the automorphism $\tau((a,b))=(a,a+b)$ induces the $\CX$ gate.

The quadratic phase gates $S_\xi$ are parameterized by \textit{quadratic forms} $\xi:G\to \CC^\times$. That is, maps $\xi$ such that the assignment 

$$g,h\mapsto \frac{\xi(gh)}{\xi(g)\xi(h)}$$

is a character in each component individually. They act by

\begin{align*}
S_\xi: \CC[G]&\to\CC[G].\\
\left| g\right>&\mapsto \xi(g)\left|g\right>
\end{align*}

The basic theory of quadratic forms is found in \cite{wall1963quadratic}. It should be emphasized that we do not require our forms to be homogenous. That is, it can be the case that $\xi(g^{-1})\neq\xi(g)$. It shouldn't be seen as a surprise that quadratic forms appear in our basic gates. Quadratic forms are ubiquitous within the study of abelian anyon models, appearing in multiple seemingly unrelated ways \cite{basak2015indicators, galindo2016solutions, wang2020and}.

The final type of Clifford gate is the Fourier transform. The Fourier transform is the map

\begin{align*}
\FF:\CC[G]&\to\CC[\Ghat]\\
\left|g\right>&\mapsto\frac{1}{\sqrt{|G|}}\sum_{\chi\in \Ghat}\overline{\chi(g)}\left|\chi\right>
\end{align*}

An issue is that the target is $\CC[\Ghat]$, not $\CC[G]$. For this reason one must parameterize the Fourier transform by some isomorphism $i:\Ghat\xrightarrow{\sim} G$, defining

\begin{align*}
\FF_i:\CC[G]&\to\CC[G].\\
\left|g\right>&\mapsto\frac{1}{\sqrt{|G|}}\sum_{\chi\in \Ghat}\overline{\chi(g)}\left|i(\chi)\right>.
\end{align*}

The formula for the inverse Fourier transform is

\begin{align*}
\FF_i^\dagger:\CC[G]&\to\CC[G],\\
\left|g\right>&\mapsto\frac{1}{\sqrt{|G|}}\sum_{\chi\in \Ghat}\chi(g)\left|\chi\circ i \right>.
\end{align*}

where we identify $\chi \circ i$ with a group element via the double dual isomorphism. We observe that $\chi\circ i=i(\chi)$ if and only if for all $\chi'\in \Ghat$ the equation $\chi'(i(\chi))=\chi(i(\chi'))$ holds, which a-priori it does not.

\begin{proposition}
\label{clifford-prop}
Let $G$ be a finite abelian group. For all $g\in G$, $\chi\in\Ghat$ we have that

\begin{enumerate}
\item $A_\tau$ is a Clifford gate, for automorphisms $\tau:G\xrightarrow{\sim} G$. Explicitly,

$$A_\tau X_g A_\tau^\dagger=X_{\tau(g)},$$  $$A_\tau Z_\chi  A_\tau^{\dagger}= Z_{\chi\circ \tau^{-1}}.$$

\item $S_\xi$ is a Clifford gate, for quadratic forms $\xi:G\to \CC^\times$. Explicitly,

$$S_\xi X_g S_\xi^\dagger = \xi(g)X_g Z_{b_\xi(g,-)},$$ $$ S_\xi Z_\chi S_\xi^\dagger = Z_\chi .$$

where $b_\xi(g,h)=\xi(gh)/(\xi(g)\xi(h))$

\item $\FF_{i}$ is a Clifford gate, for isomorphisms $i:\Ghat\xrightarrow{\sim}G$. Explicitly,

$$\FF_i X_g \FF_i^\dagger=Z_{\overline{g\circ i^{-1}}},$$ $$\FF_i Z_\chi\FF_i^\dagger = X_{i(\chi)},$$

where $\chi \circ i \in \widehat{\Ghat}$ is identified canonically with an element of $G$ via the double dual isomorphism.
\end{enumerate}
\end{proposition}
\begin{proof}These follow from straightforward calculations.
\end{proof}

In Appendix \ref{generating set appendix} we prove the converse:

\begin{theorem}\label{main-theorem-1} Let $G$ be a finite abelian group. Every Clifford gate on $G$ can be constructed from $\polylog |G|$ gates of the form $A_\tau$, $S_\xi$, and $\FF_i$.
\end{theorem}

We now discuss multi-qudit gates, i.e, $\Cliff(G^n)$. The main fact is that there is a decomposition into one-qudit gates and $\CX$ just like before. We define our generalized $\CX$ gate as

\begin{align*}
\CX: \CC[G\times G]&\xrightarrow{}\CC[G\times G].\\
\left| g\right> \left| h \right> & \mapsto \left |g \right> \left| gh \right>
\end{align*}

This gate is visualized as a circuit via

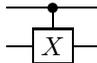
\[
\begin{array}{c}
\Qcircuit @C=1em @R=.7em {
& \ctrl{1} &\qw \\
& \gate{X} & \qw
}
\end{array}
\]

The following proposition holds:

\begin{proposition} Let $G$ be a finite abelian group. The $\CX$ gate is Clifford. Explicitly,

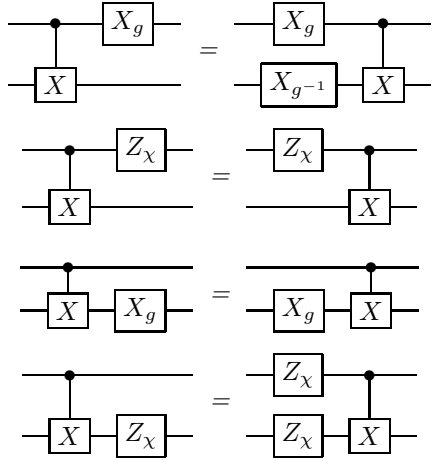
\[
\Qcircuit @C=1em @R=.7em {
   & \ctrl{1} &\gate{X_g}&\qw  & \raisebox{-2.2em}{=}  & & \gate{X_g} &\ctrl{1}&\qw\\
   & \gate{X} & \qw & \qw & & & \gate{X_{g^{-1}}} & \gate{X} & \qw &
}   
\]

\[
\Qcircuit @C=1em @R=.7em {
   & \ctrl{1} &\gate{Z_\chi}&\qw  & \raisebox{-2.2em}{=}  & & \gate{Z_\chi} &\ctrl{1}&\qw\\
   & \gate{X} & \qw & \qw & & & \qw & \gate{X} & \qw &
}   
\]

\[
\Qcircuit @C=1em @R=.7em {
   & \ctrl{1} &\qw &\qw  & \raisebox{-2.2em}{=}  & & \qw &\ctrl{1}&\qw\\
   & \gate{X} & \gate{X_g} & \qw & & & \gate{X_g} & \gate{X} & \qw &
}   
\]

\[
\Qcircuit @C=1em @R=.7em {
   & \ctrl{1} &\qw&\qw  & \raisebox{-2.2em}{=}  & & \gate{Z_\chi} &\ctrl{1}&\qw\\
   & \gate{X} & \gate{Z_\chi} & \qw & & & \gate{Z_\chi} & \gate{X} & \qw &
}   
\]

\end{proposition}
\begin{proof} This is a straightforward calculation.
\end{proof}

We now prove our main theorem about multi-qubit gates:

\begin{theorem}\label{main-theorem-2} Every Clifford operation in $\Cliff\left(G^n\right)$ can be decomposed into $\polylog|G|\cdot n$ one-qudit Clifford gates and two-qudit automorphism gates.
\end{theorem}
\begin{proof} In light of Theorem \ref{main-theorem-1}, we only need to demonstrate the result for the $A_\tau$, $S_\xi$, and $\FF_i$ gates. Choosing an appropriate isomorphism $i$ the Fourier transform on $\CC[G^n]$ is the tensor product of the Fourier transform on components. Thus, the global Fourier transform decomposes into one-qudit Clifford gates because the Fourier transform on each component is itself a Fourier transform, and hence is Clifford by Proposition \ref{clifford-prop}. Similarly, quadratic forms are still quadratic forms when restricted to subgroups and thus quadratic form gates can be decomposed into one-qudit Cliffords as well.

The last case to check is automorphism gates. This amounts to asking whether every automorphisms of $G^n$ can be written as a composition of automorphisms which affect only two of the components. To do this we appeal to the matrix theory of automorphisms, described in detail in \cite{hillar2007automorphisms} and used in Section \ref{generating set appendix} to prove Theorem \ref{main-theorem-1} . Every automorphism can be embedded into a matrix group with modular entries. It is straightforward to show that every invertible matrix can be row-reduced to the identity matrix using Gaussian elimination. Each such row reduction touches at most two rows, and hence can be implemented by an automorphism which touches at most two components. Hence, every automorphism can be decomposed as an iterated composition of automorphisms which touch at most two components so we are done.
\end{proof}

We observe, however, that Proposition \ref{main-theorem-2} does not imply every Clifford operation can be turned into one-qudit gates and $\CX$ gates. In particular, there are automorphism gates which cannot be decomposed into one-qudit automorphisms and $\CX$ automorphisms. This is in sharp contrast to the case when $G=\ZZ_d$ is cyclic, where every Clifford can be generated by one-qudit gates and $\CX$ only. The first example of this failure happens when $G=\ZZ_2\times \ZZ_4$, as shown below:

\begin{proposition}\label{counter-example} Let $G=\ZZ_2\times \ZZ_4$. Consider the subgroup $H$ of $\Aut(G^2)$ generated by automorphisms which act by the identity on $G\times \{0\}$ or $\{0\}\times G$, along with the automorphisms $(g,h)\mapsto (g, g+h)$, $(g,h)\mapsto (g+h,h)$. The subgroup $H$ is a proper subgroup. In particular, the map $((a_0,b_0),(a_1,b_1))\mapsto ((a_0,b_0),(a_0+a_1,b_1))$ is not contained in $H$.
\end{proposition}
\begin{proof} All of the automorphisms $\psi\in H$ satisfy the property that the projection of $\psi((a,0),(b,0))$ onto $\ZZ_2\times\ZZ_2\subset G\times G$ is congruent modulo 2 to the projection of  $\psi((0,a),(0,b))$ onto $\ZZ_4\times\ZZ_4\subset G\times G$.The morphism $((a_0,b_0),(a_1,b_1))\mapsto ((a_0,b_0),(a_0+a_1,b_1))$ does not have this property and hence it is not contained in $H$.
\end{proof}

We can now end with the main theorem quantifying the power of Clifford circuits:

\begin{theorem} Every circuit consisting of Clifford gates can be efficiently simulated by classical computers.
\end{theorem}
\begin{proof} It was proved in \cite{nest2012efficient} that every circuit consisting of automorphism gates, quadratic phase gates, and Fourier gates can be efficiently simulated by classical computers. This result was extended to \cite{bermejo2012classical}, which allowed the quantum computer to perform intermediate measurements. Thus, this result is an immediate corollary of Theorem \ref{main-theorem-1}.
\end{proof}

\section{Implementing the $\CX$ gate}
\label{CX gate}

In lattice surgery, the fusion of two chunks of surface code results in a two-qubit Pauli measurement \cite{horsman2012surface}. These are the only native multi-qudit operations available. Hence, one must decompose multi-qudit gates into two-qudit measurements and one-qudit Cliffords. We show in this section that such an algorithm exists for the implementation of the $\CX$ gate. However, in light of Proposition \ref{counter-example} more work needs to be done for finding such circuits for other two-qudit Clifford gates. We give the $\CX$ circuit below, using an ancilla:

\begin{proposition} The following equality holds, up to global phase:

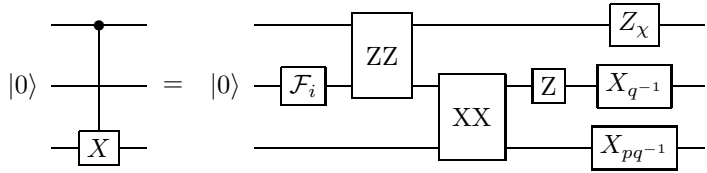
\[
\Qcircuit @C=1em @R=.7em {
   & \ctrl{2} &\qw  & & & & & \qw& \multigate{1}{\mathrm{ZZ}}&\qw & \qw & \gate{Z_{\chi}} & \qw\\
   \lstick{\ket{0}}& \qw & \qw & =& & &\lstick{\ket{0}}& \gate{\FF_i}  & \ghost{\mathrm{ZZ}} &\multigate{1}{\mathrm{XX}}& \gate{\mathrm{Z}} & \gate{X_{q^{-1}}} & \qw\\
   & \gate{X}  & \qw & & & & & \qw &\qw & \ghost{\mathrm{XX}} & \qw & \gate{X_{pq^{-1}}} & \qw\\
}   
\]

where $p$ is the result of the ZZ measurement, $\chi$ is the result of the XX measurement, and $q$ is the result of the Z measurement.

\end{proposition}
\begin{proof} We show that $\ket{g}\ket{0}\ket{h}$ sends to $\ket{g}\ket{0}\ket{gh}$, up to a global phase independent of $g$ and $h$. After applying a Fourier transform, the state is $|G|^{-1/2}\sum_{k}\ket{g}\ket{k}\ket{h}$. Measuring $p$ under ZZ means we project orthogonally onto the subspace with $p=gk$. This results in the state $\ket{g}\ket{pg^{-1}}\ket{h}$. Applying an XX measurement gives the state

$$\frac{1}{\sqrt{|G|}}\sum_{k\in G}\overline{\chi(k)}\ket{g}\ket{pg^{-1}k}\ket{hk}.$$

Measuring Z projects onto the subspace $q=pg^{-1}k$, yielding the state $\overline{\chi(qp^{-1}g)}\ket{g}\ket{q}\ket{hqp^{-1}g}$. Applying the final corrections yields the state $\overline{\chi(qp^{-1})}\ket{g}\ket{0}\ket{hg}$. Since the phase $\overline{\chi(qp^{-1})}$ is independent of $g$ and $h$, we have proved the desired result.

\end{proof}

\section{Magic states}
\label{magic-states}

In this section we consider the prospect of achieving universal quantum computing using magic states. Given access to special magic states, certain non-Clifford gates can be implemented with Cliffords. The main circuit of magic state theory is shown below:

\begin{proposition}\label{magic-states} Let $\xi:G\to U(1)$ be any function. Let $S_\xi:\CC[G]\to\CC[G]$ be the operator $S_\xi\ket{g}=\xi(g)\ket{g}$. The following equality holds up to phase:

\[
\Qcircuit @C=1em @R=.7em {
   & \ctrl{1}&\qw & \gate{S_{\overline{b_{\xi}(k,-)}}}& \qw  & \raisebox{-2.2em}{=}  & & \gate{S_{\xi}} &\qw\\
   \lstick{\frac{1}{\sqrt{|G|}}\sum_{g\in G}\xi(g)\ket{g}}& \gate{X^{\dagger}} & \meter & \qw & \qw & & & \qw & \qw &
}   
\]

where $k$ is the result of the measurement on the bottom qubit, and $b_\xi(k,g)=\xi(kg)/(\xi(k)\xi(g))$.
\end{proposition}

\begin{proof} Letting the top wire input the state $\ket{g}$, the state after applying $\CX^\dagger$ is

$$\frac{1}{\sqrt{|G|}}\sum_{h\in G}\xi(h)\ket{g}\ket{g^{-1}h}.$$

Measuring $k=g^{-1}h$ yields the state $\xi(gk)\ket{g}\ket{k}$. Applying $S_{\overline{b_\xi(k,-)}}$ thus yields $\xi(g)\xi(k)\ket{g}\ket{k}$. Since $\xi(k)$ is a global phase, we thus get the desired result.
\end{proof}

In particular, suppose that $\xi:G\to U(1)$ is a function such that $b_\xi(k,-)$ is a quadratic form for every $k\in G$. Then, with a supply of magic states of the form $\frac{1}{\sqrt{|G|}}\sum_{g\in G}\xi(g)\ket{g}$. One can use the circuit of Proposition \ref{magic-states} to perform the gate $S_\xi$ using only Clifford gates. It is a great result that this is enough for universal quantum computing:

\begin{proposition} If $\xi:G\to U(1)$ is a function such that $b_\xi(k,-)$ is a quadratic form for every $k\in G$ and $\xi$ is itself not a quadratic form, then Clifford gates along with single-qubit $S_\xi$ gates are enough for universal quantum computing. 
\end{proposition}
\begin{proof} Since $\xi:G\to U(1)$ is not a quadratic form on $G$, then there must exist a cyclic subgroup of $G$ such that $\xi$ is not a quadratic form on that subgroup. Thus, we can reduce to the case that $G=\ZZ_n$. This result now follows from well-known results about qubits \cite{wang2020qudits}.
\end{proof}

The existence of suitable $\xi$ is immediate. For cyclic groups one can explicitly write $\xi$ as being an exponential with a cubic exponent, and for more general finite abelian groups one can take products of such $\xi$. A good discussion of gates of this form as well as gates higher in the so-called ``Clifford hierarchy" for qudits is found in \cite{cui2017diagonal, rengaswamy2019unifying}. The states needed to apply the circuit in Proposition \ref{magic-states} can be created using the process of magic state distillation \cite{bravyi2005universal, litinski2019magic}.

\section{Fourier transform on split subgroups}
\label{fourier-transform}

One useful circuit in $\Cliff(G)$ is the implementation of Fourier transforms on split subgroups of $G$. That is, a circuit implementing the Fourier transform on the $H$ term of some splitting $G=H\times K$, using only $G$-Clifford gates.

To describe our circuit, let $\xi: G\to\CC^\times$ be a non-degenerate quadratic form. Given $\chi\in \Ghat$ let $i_\xi(\chi)$ be the unique element such that

$$\frac{\xi(g)\xi(i_\xi(\chi))}{\xi(gi_\xi(\chi))}=\chi(g),\,\, \forall g\in G.$$

\begin{theorem}\label{fourier-decomposition} Let $\xi:G\to\CC^\times$ be a non-degenerate quadratic form.

\begin{enumerate}
\item $(F_{i_\xi}S_\xi)^3= I$ up to global phase

\item $F_{i_\xi}^2=A_{(\--)^{-1}}$, where $A_{(\--)^{-1}}$ is the automorphism gate associated to the inversion automorphism, $g\mapsto g^{-1}$.

\item Let $H$ be another finite group, and let $i: \widehat{H}\xrightarrow{\sim}H$ be any isomorphism. Letting the top wire denote the Hilbert space $\CC[G]$ and letting the bottom wire denote $\CC[H]$, we have that

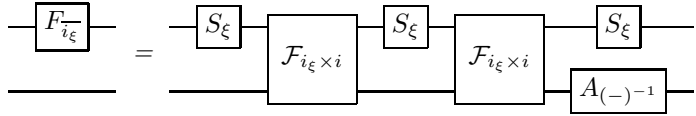
\[
\begin{array}{c}
\Qcircuit @C=1em @R=.7em {
& \gate{F_{\overline{i_\xi}}} &\qw &\raisebox{-2.2em}{=} && \gate{S_\xi} &\multigate{1}{\FF_{i_\xi \times i}} & \gate{S_\xi} & \multigate{1}{\FF_{i_\xi \times i}} & \gate{S_\xi} & \qw\\
&\qw &\qw & & & \qw  & \ghost{\FF_{i_\xi \times i}} & \qw  & \ghost{\FF_{i_\xi \times i}} & \gate{A_{(\--)^{-1}}} & \qw
}
\end{array}
\]

\end{enumerate}
\end{theorem}
\begin{proof} We begin by proving part (1). Choose $g\in G$. We compute

\begin{align*}
\left(\FF_{i_\xi}S_\xi\right)^3\ket{g}&=\frac{\FF_{i_\xi}S_\xi\FF_{i_\xi}}{\sqrt{|G|}}\sum_{\chi\in \Ghat}\xi(i_\xi(\chi))\xi(g)\overline{\chi(g)}\ket{i_\xi(\chi)}.
\end{align*}

By the defining formula of $i_\xi(\chi)$ we get that $\xi(i_\xi(\chi))\xi(g)\overline{\chi(g)}=\xi(gi_\xi(\chi))$. Hence, changing variables via $h=gi_{\xi}(\chi)$ we arrive at

\begin{align*}
\left(\FF_{i_\xi}S_\xi\right)^3\ket{g}&=\frac{\FF_{i_\xi}S_\xi\FF_{i_\xi}}{\sqrt{|G|}}\sum_{h\in G}\xi(h)\ket{g^{-1} h}.
\end{align*}

Expanding further, we obtain

\begin{align*}
\left(\FF_{i_\xi}S_\xi\right)^3\ket{g}&=\frac{\FF_{i_\xi}}{|G|}\sum_{\chi\in\Ghat}\sum_{h\in G}\xi(h)\xi(i_\xi(\chi))\overline{\chi(h)}\chi(g)\ket{i_\xi(\chi)}.
\end{align*}

Applying the defining relation $\xi(h)\xi(i_\xi(\chi))\overline{\chi(h)}=\xi(hi_\xi(\chi))$ and changing variables via $h\mapsto h \cdot \overline{i_{\xi}(\chi)}$ gives

\begin{align*}
\left(\FF_{i_\xi}S_\xi\right)^3\ket{g}&=\left(\frac{1}{\sqrt{|G|}}\sum_{h\in G}\xi(h)\right)\cdot\frac{\FF_{i_\xi}}{\sqrt{|G|}}\sum_{\chi\in\Ghat}\chi(g)\ket{i_\xi(\chi)}\\
&=\left(\frac{1}{\sqrt{|G|}}\sum_{h\in G}\xi(h)\right)\ket{g}.
\end{align*}

Since $\ket{g}$ was chosen arbitrarily and the phase has no dependence on $g$, this proves our result. Part (2) is obvious, and part (3) immediately follows from the first two.
\end{proof}

The equations $(FS)^3=I$ and $F^2=A_{(\--)^{-1}}$ are shadows of a deeper mathematical truth: topological quantum field theories all have associated modular representations \cite{wang2010topological}. That is, every topological quantum field theory has an associated representation of the modular group $\mathrm{SL}_2(\ZZ)$. By the canonical presentation for $\mathrm{SL}_2(\ZZ)$, this means that gives a pair of matrices $(S,T)$ such that $(ST)^3=I$ and $T^2=I$ up to phase. Every finite group has an associated topological quantum field theory. It turns out that $S_\xi$ is essentially to $T$-matrix, and $F_{i_\xi}$ is a half-twisted version of the $S$-matrix. The half twist is neccecary to make so $T^2=I$ holds, as opposed to simply $T^2=A_{(-)^{-1}}$.

\section{The generating set of Clifford gates}
\label{generating set appendix}

The aim of this section is to prove Theorem \ref{main-theorem-1}, which asserts that every $G$-Clifford gate can be decomposed into Fourier, quadratic phase, and automorphism gates. The main idea of the proof is as follows. Every Clifford operator $U$ induces a map $\Pauli(G)/U(1)\xrightarrow{}\Pauli(G)/U(1)$ by  sending $P$ to $UPU^{\dagger}$. Seeing as $\Pauli(G)/U(1)\cong G\times \Ghat$, we get a natural map $\Cliff(G)\to\Aut(G\times \Ghat)$. To understand $\Cliff(G)$ we study the properties of this map. Our main tool will be the symplectic group on $G$:

$$\Sp(G)=\{\sigma\in\Aut(G\times\Ghat)\st \beta(\sigma(x),\sigma(y))=\beta(x,y)\},$$

where $\beta$ is the bilinear form

\begin{align*}
\beta:(G\times \Ghat)\times (G\times \Ghat)&\to \CC^\times\\
(g_0,\chi_0)\times (g_1,\chi_1)&\mapsto \chi_0(g_1)\overline{\chi_1(g_0)}.
\end{align*}

Our first step towards proving Theorem \ref{main-theorem-1} is the following:

\begin{proposition}\label{short-exact} The image of the natural map $\Cliff(G)\xrightarrow{}\Aut(G\times \Ghat)$ is contained in $\Sp(G)$, and the sequence

$$0\to \Pauli(G)\to \Cliff(G)\to\Sp(G)$$

is exact.
\end{proposition}
\begin{proof} To begin, we observe that for every pair of pairs $(g_0,\chi_0), (g_1,\chi_1)\in G\times \Ghat$ we have that

$$(X_{g_0}Z_{\chi_0})(X_{g_1}Z_{\chi_1})=\chi_0(g_1)\overline{\chi_1(g_0)}(X_{g_1}Z_{\chi_1})(X_{g_0}Z_{\chi_0}).$$

Thus, we can re-cast $\beta((g_0,\chi_0),(g_1,\chi_1))$ as the relative phase that appears upon commuting $X_{g_0}Z_{\chi_0}$ and $X_{g_1}Z_{\chi_1}$. Conjugating by $U\in\Cliff(G)$ must preserve this relative phase. Thus, the image of $U$ in $\Aut(G\times \Ghat)$ must preserve $\beta(x,y)$, and thus the image of $U$ must lie in $\Sp(G)$.

The kernel of the map $\Cliff(G)\to \Sp(G)$ consists of all matrices which act trivially on $\Pauli(G)/U(1)$. That is, matrices which commute with $\Pauli(G)$ up to phase. By Proposition \ref{clifford-prop}, these are exactly the Pauli matrices. Thus, our proof is complete.
\end{proof}

The proof strategy is now as follows. Representing automorphisms as matrices, we will go through an explicit set of row-reduction operations to show that every matrix in $\Sp(G)$ can be decomposed into the matrices $A_\tau$, $S_\xi$, $\FF_i$. This wll show that the subgroup generated by the $A_\tau$, $S_\xi$, $\FF_i$ surjects onto $\Sp(G)$. Finally, from Proposition \ref{short-exact} we can conclude the result. Before proceeding, we introduce the group of symmetric bilinear forms,

$$\mathrm{SymBil}(G)=\{b:G\times G\to \CC^\times \st b(x,y)=b(y,x)\}.$$

Every symmetric bilinear form induces a map $\tilde{b}:G\to \Ghat$ sending $g$ to $b(g,-)$. This map satisfies $\tilde{b}(x)(y)=\tilde{b}(y)(x)$. Similarly, every symmetric homomorphism $G\to\Ghat$ induces a symmetric bilinear form. We will conflate symmetric bilinear forms and their induced homomorphisms $G\to\Ghat$. It is immediate to verify that any homomorphism $A\times B\to \widehat{A}\times \widehat{B}$ which restricts to symmetric bilinear forms $A\to \widehat{A}$ and $B\to\widehat{B}$ must be itself a symmetric bilinear form. This gives a very efficient method for constructing symmetric bilinear forms on a product of cyclic groups, since every map between cyclic groups is symmetric.

Additionally, we have a map $\mathrm{Quad}(G)\to\mathrm{SymBil}(G)$ defined by sending $\xi:G\to \CC^\times$ to $b(g,h)=\xi(gh)/(\xi(g)\xi(h))$. This leads to the following key tool:

\begin{lemma}\label{short-exact-mini} For any finite group $G$, we have a short exact sequence

$$0\to \Ghat \to \mathrm{Quad}(G)\to \mathrm{SymBil}(G)\to0,$$

where $\Ghat$ is the character group of $G$, $\mathrm{Quad}(G)$ is the group of quadratic forms, and $\mathrm{SymBil}(G)$ is the group of symmetric bilinear maps $G\times G\to \CC^\times$.

\end{lemma}
\begin{proof} This is a well known result, c.f. Wall's original paper \cite{wall1963quadratic}, and Theorem 2.1 in \cite{basak2015indicators}.
\end{proof}

Proposition \ref{short-exact} can be seen as a quantum-information analogue of Lemma \ref{short-exact-mini}. We now describe the matrix theory of automorphisms, which can also be found in more detail in introductory texts such as \cite{hillar2007automorphisms}. The fundamental theorem of finite abelian groups tells us that

$$G\cong \ZZ_{q_0}\times \ZZ_{q_1}\times\ZZ_{q_2}...\times \ZZ_{q_d}$$

with $q_i | q_{i-1}$. Notationally, we write elements of $G\times\Ghat$ as vectors

\[
\begin{blockarray}{cc}
\begin{block}{c(c)}
  \ZZ_{q_0} & ...  \\
  ... & ...\\
  \ZZ_{q_d} & ... \\
\cline{2-2}
  \widehat{\ZZ}_{q_0} & ... \\
  ... & ...  \\
  \widehat{\ZZ}_{q_d} & ...  \\
\end{block}
\end{blockarray}
 \]

Additionally, we can write elements of $G\times \Ghat$ as matrices. A-priori, elements of diffrent cyclic components are incomparable. This is fixed by viewing each $\ZZ_{q_i}$ as the subgroup of $\ZZ_{q_0}$ consisting of multiples of $q_0/q_i$. For convenience with comparison between group elements and characters we canonically identify $\widehat{\ZZ}_{q_i}$ with $\ZZ_{q_i}$ by working with the distinguished generator $n\mapsto e^{2\pi i n / q_i}$.

Every homomorphism between cyclic groups is a scalar multiplication. Hence, we can compress data by writing the entries of the matrix as integers which represent scalar multiplication. For example, if $G=\ZZ_4\times \ZZ_2$ we have the matrix equation

\[
\begin{blockarray}{ccccc}
& \ZZ_{4} & \ZZ_{2} & \widehat{\ZZ}_4 & \widehat{\ZZ}_2 \\
\begin{block}{c(cc | cc)}
  \ZZ_{4} & 2 & 1 & 3 & 0\\
  \ZZ_{2} & 2 & 0 & 2  & 0\\
\cline{2-5}
  \widehat{\ZZ}_{4} & 1 & 0 & 0 & 1\\
  \widehat{\ZZ}_{2} & 0 & 1 & 0 & 1\\
\end{block}
\end{blockarray}
\begin{pmatrix}
1\\
0\\
\hline
3\\
2
\end{pmatrix}
=
\begin{pmatrix}
2\cdot 1 + 1\cdot 0 + 3\cdot 3 + 0\cdot 2\\
2\cdot 1 + 0\cdot 0 + 2\cdot 3 + 0\cdot 2 \\
\hline
 1\cdot 1 + 0\cdot 0 + 0\cdot 3 + 1\cdot 2\\
0\cdot 1 + 1\cdot 0 + 0\cdot 3 + 1\cdot 2\\
\end{pmatrix}
=
\begin{pmatrix}
3\\
0\\
\hline
3\\
2
\end{pmatrix}
 \]

We now can begin our proof:

\begin{proof}[Proof of Theorem \ref{main-theorem-1}]

We begin with a generic automorphism of $G\times \Ghat$ in $\Sp(G)$ and reduce to the identity matrix using automorphisms of the form $A_\tau$, $S_\xi$, and $\FF_i$. We proceed by induction on the number of cyclic factors in the canonical decomposition of $G$. Suppose that we can always reduce elements of $\Sp(G)$ to matrices of the form

\[
\begin{blockarray}{ccccccc}
& \ZZ_{q_0} & ... & \ZZ_{q_d} & \widehat{\ZZ}_{q_0} & ... & \widehat{\ZZ}_{q_d} \\
\begin{block}{c(c | cc | c | cc)}
  \ZZ_{q_0} & 1 & 0 & 0 & 0 & 0 & 0 \\
\cline{2-7}
  ... & 0 & ... & ... & 0 & ... & ... \\
  \ZZ_{q_d} & 0 & ... & ...  & 0 & ... & ...\\
\cline{2-7}
  \widehat{\ZZ}_{q_0} & 0 & 0 & 0 & 1 & 0 & 0 \\
\cline{2-7}
  ... & 0 & ... & ... & 0 & ... & ... \\
  \widehat{\ZZ}_{q_d} & 0 & ... & ... & 0 & ... & ... \\
\end{block}
\end{blockarray}
 \]

Then, further row reductions can be performed on the factors $\ZZ_{q_1}...\ZZ_{q_d}$ without affecting the $\ZZ_{q_0}$-entries. This is true for $A_\tau$ and $S_\xi$ gates because applying automorphisms and phases on one collection of cyclic factors does not affect the others, and it is true for $\FF_i$ gates because we can apply partial Fourier transforms by Proposition \ref{fourier-decomposition}. In particular, by induction we would arrive at a proof of the result. Thus, the proof amounts to finding a series of reductions to arrive at a matrix of this form. We label our generic matrix as below:

\[
\begin{blockarray}{ccccccc}
& \ZZ_{q_0} & ... & \ZZ_{q_d} & \widehat{\ZZ}_{q_0} & ... & \widehat{\ZZ}_{q_d} \\
\begin{block}{c(c | cc | c | cc)}
  \ZZ_{q_0} & r_0 & ... & ... & ... & ... & ... \\
\cline{2-7}
  ... & ... & ... & ... & ... & ... & ... \\
  \ZZ_{q_d} & r_d & ... & ...  & ... & ... & ...\\
\cline{2-7}
  \widehat{\ZZ}_{q_0} & r_0' & ... & ... & ... & ... & ... \\
\cline{2-7}
  ... & ... & ... & ... & ... & ... & ... \\
  \widehat{\ZZ}_{q_d} & r_d' & ... & ... & ... & ... & ... \\
\end{block}
\end{blockarray}
 \]

We start by applying $S_\xi$ operators. By the short exact sequence of Lemma \ref{short-exact-mini}, any homomorphism $b:G\to \Ghat$ with $b(x)(y)=b(y)(x)$ can be induced as $g\mapsto b_\xi(g,-)=\xi(g\cdot - )/(\xi(g)\xi(-))$. We construct such a map $b$ now. By Bezout's lemma, there exists integers $k_i$ such that $r_i+k_ir_i'=\gcd(r_i,r_i')$ in $\ZZ_{q_i}$. Let $b$ be the map which acts as multiplication by $k_i$ on the $\ZZ_{q_i}$ component. Postcomposing with $b$ and using the commutation relation in Proposition \ref{clifford-prop}, we reduce our matrix to

\[
\begin{blockarray}{ccccccc}
& \ZZ_{q_0} & ... & \ZZ_{q_d} & \widehat{\ZZ}_{q_0} & ... & \widehat{\ZZ}_{q_d} \\
\begin{block}{c(c | cc | c | cc)}
  \ZZ_{q_0} & r_0 & ... & ... & ... & ... & ... \\
\cline{2-7}
  ... & ... & ... & ... & ... & ... & ... \\
  \ZZ_{q_d} & r_d & ... & ...  & ... & ... & ...\\
\cline{2-7}
  \widehat{\ZZ}_{q_0} & \gcd(r_0,r_0') & ... & ... & ... & ... & ... \\
\cline{2-7}
  ... & ... & ... & ... & ... & ... & ... \\
  \widehat{\ZZ}_{q_d} & \gcd(r_d,r_d') & ... & ... & ... & ... & ... \\
\end{block}
\end{blockarray}
 \]

Now, let $i:G\xrightarrow{\sim} \Ghat$ be the isomorphism which acts by $1$ on each component. Postcomposing with the Fourier transform $\FF_i$ has the effect of swapping the $\ZZ_{q_i}$ and $\widehat{\ZZ}_{q_i}$ components by Proposition \ref{clifford-prop}. Thus, we arrive at 

\[
\begin{blockarray}{ccccccc}
& \ZZ_{q_0} & ... & \ZZ_{q_d} & \widehat{\ZZ}_{q_0} & ... & \widehat{\ZZ}_{q_d} \\
\begin{block}{c(c | cc | c | cc)}
  \ZZ_{q_0} & \gcd(r_0,r_0') & ... & ... & ... & ... & ... \\
\cline{2-7}
  ... & ... & ... & ... & ... & ... & ... \\
  \ZZ_{q_d} & \gcd(r_d,r_d') & ... & ...  & ... & ... & ...\\
\cline{2-7}
  \widehat{\ZZ}_{q_0} & r_0& ... & ... & ... & ... & ... \\
\cline{2-7}
  ... & ... & ... & ... & ... & ... & ... \\
  \widehat{\ZZ}_{q_d} & r_d & ... & ... & ... & ... & ... \\
\end{block}
\end{blockarray}
 \]

Now, since the matrix is an automorphism the greatest common divisor of the first column must be $1$. Hence, the greatest common divisor of $\gcd(r_0,r_0')...\gcd(r_d,r_d')$ is $1$. Using the fact that $q_0$ is an element of maximal order in $G$, we know there exists an automorphism $\tau$ of $G$ which maps $(\gcd(r_0,r_0'),...\gcd(r_d,r_d'))$ to $(1,0,0..,0)$. Applying $A_\tau$ and usiring the commutation relationship of Proposition \ref{clifford-prop}, we reduce our matrix to

\[
\begin{blockarray}{ccccccc}
& \ZZ_{q_0} & ... & \ZZ_{q_d} & \widehat{\ZZ}_{q_0} & ... & \widehat{\ZZ}_{q_d} \\
\begin{block}{c(c | cc | c | cc)}
  \ZZ_{q_0} & 1 & ... & ... & ... & ... & ... \\
\cline{2-7}
  ... & 0 & ... & ... & ... & ... & ... \\
  \ZZ_{q_d} &0 & ... & ...  & ... & ... & ...\\
\cline{2-7}
  \widehat{\ZZ}_{q_0} & r_0'' & ... & ... & ... & ... & ... \\
\cline{2-7}
  ... & ... & ... & ... & ... & ... & ... \\
  \widehat{\ZZ}_{q_d} & r_d'' & ... & ... & ... & ... & ... \\
\end{block}
\end{blockarray}
 \]

Now, define the homomorphism $b:G\to \Ghat$ by sending $(1,0,0..,0)$ to $(-r_0'',-r_1''...,-r_d'')$ and acting by $0$ on all other components. Using Lemma \ref{short-exact-mini} we know that $b$ is induced by a quadratic form $\xi$. Appling the corresponding $S_\xi$ gate, we reduce without loss of generality to

\[
\begin{blockarray}{ccccccc}
& \ZZ_{q_0} & ... & \ZZ_{q_d} & \widehat{\ZZ}_{q_0} & ... & \widehat{\ZZ}_{q_d} \\
\begin{block}{c(c | cc | c | cc)}
  \ZZ_{q_0} & 1 & ... & ... & ... & ... & ... \\
\cline{2-7}
  ... & 0 & ... & ... & ... & ... & ... \\
  \ZZ_{q_d} &0 & ... & ...  & ... & ... & ...\\
\cline{2-7}
  \widehat{\ZZ}_{q_0} & 0 & ... & ... & r & ... & ... \\
\cline{2-7}
  ... & 0 & ... & ... & ... & ... & ... \\
  \widehat{\ZZ}_{q_d} & 0 & ... & ... & ... & ... & ... \\
\end{block}
\end{blockarray}
 \]

We now apply the symplectic condition of $\Sp(G)$ to the pair of pairs $x=(g_0,0)$ and $y=(0,\chi_1)$, where $g_0=(1,0,0,...0)$, $\chi_1=(1,0,0,...0)$. It gives us that $r=b(\sigma(x),\sigma(y))=b(x,y)=1$. Hence, $r=1$. We now repeat the same procedure as before but dualized, \textit{postcomposing} with the appropriate $A_\tau$ and $S_\xi$ gates. This reduces us to

\[
\begin{blockarray}{ccccccc}
& \ZZ_{q_0} & ... & \ZZ_{q_d} & \widehat{\ZZ}_{q_0} & ... & \widehat{\ZZ}_{q_d} \\
\begin{block}{c(c | cc | c | cc)}
  \ZZ_{q_0} & 1 & ... & ... & 0 & ... & ... \\
\cline{2-7}
  ... & 0 & ... & ... & 0 & ... & ... \\
  \ZZ_{q_d} &0 & ... & ...  & 0 & ... & ...\\
\cline{2-7}
  \widehat{\ZZ}_{q_0} & 0 & ... & ... & 1 & ... & ... \\
\cline{2-7}
  ... & 0 & ... & ... & 0 & ... & ... \\
  \widehat{\ZZ}_{q_d} & 0 & ... & ... & 0 & ... & ... \\
\end{block}
\end{blockarray}
 \]

It is simple to check that such postcomposition does not affect the entries in the first column, and thus that our reduction is valid. We now apply the symplecticity condition with the pair of pairs $x=(0,\chi_0)$ and $y=(g_1,\chi_1)$, where $\chi_0=(1,0,0,...0)$. Varying $g_1$ and $\chi_1$, we conclude that the first row has all zeros after the 1st entry. Similarly, applying the symplecticity condition with $x=(g_0,0)$, $g_0=(1,0,0,...0)$ we conclude that the $\widehat{\ZZ}_{q_0}$ row has all zeros outside of the $(\widehat{\ZZ}_{q_0},\widehat{\ZZ}_{q_0})$ entry. Thus, we arrive at the reduction

\[
\begin{blockarray}{ccccccc}
& \ZZ_{q_0} & ... & \ZZ_{q_d} & \widehat{\ZZ}_{q_0} & ... & \widehat{\ZZ}_{q_d} \\
\begin{block}{c(c | cc | c | cc)}
  \ZZ_{q_0} & 1 & 0 & 0 & 0 & 0 & 0 \\
\cline{2-7}
  ... & 0 & ... & ... & 0 & ... & ... \\
  \ZZ_{q_d} &0 & ... & ...  & 0 & ... & ...\\
\cline{2-7}
  \widehat{\ZZ}_{q_0} & 0 & 0 & 0 & 1 & 0 & 0 \\
\cline{2-7}
  ... & 0 & ... & ... & 0 & ... & ... \\
  \widehat{\ZZ}_{q_d} & 0 & ... & ... & 0 & ... & ... \\
\end{block}
\end{blockarray}
 \]

By induction on the number of cyclic factors in $G$, we conclude the result.
\end{proof}

\section{Conclusion and future directions}
\label{conclusion}

One interesting question is to what extent the difference between $G$-Cliffords for different $G$ can be exploited. For instance, the primitive $\ZZ_8$-Pauli $Z_\chi$ operator gives the $T$-gate on a subset. Hence, $\ZZ_2$-Clifford + $\ZZ_8$-Pauli is universal. Going even further, it has been shown that $\ZZ_2$-Clifford+$\ZZ_4$-Clifford is universal in \cite{moussa2015quantum}. A natural question is whether or not one can glue together two abelian phases along some domain wall to exploit this connection for a universal quantum computation scheme. One unsuccessful attempt in this direction was made in \cite{moussa2016transversal}. A motivation for the fact that this could be possible is that excitations in domain boundaries of abelian phases are non-abelian, and hence when new quantum computing primitives such as projective charge measurement are allowed one can arrive at universal quantum computation in an abelian phase \cite{cong2016topological, cong2017universal}. The theory of domain walls is now well established \cite{kesselring2018boundaries, scruby2020hierarchy, bombin2023logical}, and placing two different phases together to simplify universal quantum computation has been successfully performed. Specifically, it has been shown that the universal power of $S_3$ anyons can be passed through a domain wall to $\ZZ_2$ anyons, where the simpler theory can be used to store and manipulate the quantum information \cite{laubscher2019universal}. It is unclear whether or not such a construction can be made universal purely with abelian anyons.

While there are a medley of approaches to avoiding no-go theorems for universal quantum computation with abelian anyons \cite{bravyi2005universal, paetznick2013universal, bombin2015gauge, brown2020fault}, one of the most interesting approaches still is to go to non-abelian groups. Storing information in the internal degrees of freedom of excited states allows for universal quantum computation with every non-nilpotent group \cite{mochon2004anyon, cui2015universal}. One interesting open problem is to determine the compuational power of the system where quantum information is stored in the fusion space of boundary excitations instead. Stated more loosely: what is the power of the ``non-abelian surface code"? This model is poorly studied, and was only shown recently to formally give an error correcting code \cite{cui2020kitaev}. Adding any gate to the $\ZZ_2$-Clifford group allows for universal quantum computation. Expecting non-abelian groups to be strictly more powerful than abelian ones, it is natural to conjecture that even nilpotent groups could give universal quantum computing. These theories are especially relevant seeing as the quantum double model associated to the (nilpotent) dihedral group $D_4$ has been implemented on quantum processor \cite{iqbal2023creation}. One of the big hurdles towards this development is the complexity of boundaries in non-abelian models \cite{beigi2011quantum, cong2016topological}. Dealing properly with boundaries is key to a good theory of non-abelian surface codes.

Even for abelian anyons it is not trivial how to realize the most powerful gate set possible. A naive implementation would only give single qubit Clifford gates and a two-qubit CNOT gate, which as we demonstrated in Proposition \ref{counter-example} is not enough to generate the full Clifford group for non-cyclic gauge groups. Properly dealing with the various possible boundary theories seems to be the key challenge.

\bibliographystyle{ieeetr}
\bibliography{ref}

\begin{thebibliography}{10}

\bibitem{bravyi2013classification}
S.~Bravyi and R.~K{\"o}nig, ``Classification of topologically protected gates
  for local stabilizer codes,'' {\em Physical review letters}, vol.~110,
  no.~17, p.~170503, 2013.

\bibitem{eastin2009restrictions}
B.~Eastin and E.~Knill, ``Restrictions on transversal encoded quantum gate
  sets,'' {\em Physical review letters}, vol.~102, no.~11, p.~110502, 2009.

\bibitem{gottesman1998heisenberg}
D.~Gottesman, ``The heisenberg representation of quantum computers,'' {\em
  arXiv preprint quant-ph/9807006}, 1998.

\bibitem{bravyi2005universal}
S.~Bravyi and A.~Kitaev, ``Universal quantum computation with ideal clifford
  gates and noisy ancillas,'' {\em Physical Review A}, vol.~71, no.~2,
  p.~022316, 2005.

\bibitem{hostens2005stabilizer}
E.~Hostens, J.~Dehaene, and B.~De~Moor, ``Stabilizer states and clifford
  operations for systems of arbitrary dimensions and modular arithmetic,'' {\em
  Physical Review A}, vol.~71, no.~4, p.~042315, 2005.

\bibitem{brown2017poking}
B.~J. Brown, K.~Laubscher, M.~S. Kesselring, and J.~R. Wootton, ``Poking holes
  and cutting corners to achieve clifford gates with the surface code,'' {\em
  Physical Review X}, vol.~7, no.~2, p.~021029, 2017.

\bibitem{zeng2008semi}
B.~Zeng, X.~Chen, and I.~L. Chuang, ``Semi-clifford operations, structure of c
  k hierarchy, and gate complexity for fault-tolerant quantum computation,''
  {\em Physical Review A}, vol.~77, no.~4, p.~042313, 2008.

\bibitem{rengaswamy2019unifying}
N.~Rengaswamy, R.~Calderbank, and H.~D. Pfister, ``Unifying the clifford
  hierarchy via symmetric matrices over rings,'' {\em Physical Review A},
  vol.~100, no.~2, p.~022304, 2019.

\bibitem{cui2017diagonal}
S.~X. Cui, D.~Gottesman, and A.~Krishna, ``Diagonal gates in the clifford
  hierarchy,'' {\em Physical Review A}, vol.~95, no.~1, p.~012329, 2017.

\bibitem{kitaev2003fault}
A.~Y. Kitaev, ``Fault-tolerant quantum computation by anyons,'' {\em Annals of
  physics}, vol.~303, no.~1, pp.~2--30, 2003.

\bibitem{kitaev1997quantum}
A.~Y. Kitaev, ``Quantum error correction with imperfect gates,'' in {\em
  Quantum communication, computing, and measurement}, pp.~181--188, Springer,
  1997.

\bibitem{bombin2008family}
H.~Bombin and M.~Martin-Delgado, ``Family of non-abelian kitaev models on a
  lattice: Topological condensation and confinement,'' {\em Physical Review B},
  vol.~78, no.~11, p.~115421, 2008.

\bibitem{yan2022ribbon}
B.~Yan, P.~Chen, and S.~X. Cui, ``Ribbon operators in the generalized kitaev
  quantum double model based on hopf algebras,'' {\em Journal of Physics A:
  Mathematical and Theoretical}, vol.~55, no.~18, p.~185201, 2022.

\bibitem{cowtan2022quantum}
A.~Cowtan and S.~Majid, ``Quantum double aspects of surface code models,'' {\em
  Journal of Mathematical Physics}, vol.~63, no.~4, 2022.

\bibitem{anwar2014fast}
H.~Anwar, B.~J. Brown, E.~T. Campbell, and D.~E. Browne, ``Fast decoders for
  qudit topological codes,'' {\em New Journal of Physics}, vol.~16, no.~6,
  p.~063038, 2014.

\bibitem{bullock2007qudit}
S.~S. Bullock and G.~K. Brennen, ``Qudit surface codes and gauge theory with
  finite cyclic groups,'' {\em Journal of Physics A: Mathematical and
  Theoretical}, vol.~40, no.~13, p.~3481, 2007.

\bibitem{knill1996group}
E.~Knill, ``Group representations, error bases and quantum codes,'' {\em arXiv
  preprint quant-ph/9608049}, 1996.

\bibitem{knill1996non}
E.~Knill, ``Non-binary unitary error bases and quantum codes,'' {\em arXiv
  preprint quant-ph/9608048}, 1996.

\bibitem{de2013linearized}
N.~De~Beaudrap, ``A linearized stabilizer formalism for systems of finite
  dimension,'' {\em Quantum Information \& Computation}, vol.~13, no.~1-2,
  pp.~73--115, 2013.

\bibitem{duclos2013kitaev}
G.~Duclos-Cianci and D.~Poulin, ``Kitaev's z d-code threshold estimates,'' {\em
  Physical Review A}, vol.~87, no.~6, p.~062338, 2013.

\bibitem{gottesman1998fault}
D.~Gottesman, ``Fault-tolerant quantum computation with higher-dimensional
  systems,'' in {\em NASA International Conference on Quantum Computing and
  Quantum Communications}, pp.~302--313, Springer, 1998.

\bibitem{wang2020qudits}
Y.~Wang, Z.~Hu, B.~C. Sanders, and S.~Kais, ``Qudits and high-dimensional
  quantum computing,'' {\em Frontiers in Physics}, vol.~8, p.~589504, 2020.

\bibitem{cano2014bulk}
J.~Cano, M.~Cheng, M.~Mulligan, C.~Nayak, E.~Plamadeala, and J.~Yard,
  ``Bulk-edge correspondence in (2+ 1)-dimensional abelian topological
  phases,'' {\em Physical Review B}, vol.~89, no.~11, p.~115116, 2014.

\bibitem{gottesman1999quantum}
D.~Gottesman and I.~L. Chuang, ``Quantum teleportation is a universal
  computational primitive,'' {\em arXiv preprint quant-ph/9908010}, 1999.

\bibitem{bermejo2014classical}
J.~Bermejo-Vega and M.~Van Den~Nest, ``Classical simulations of abelian-group
  normalizer circuits with intermediate measurements,'' {\em Quantum
  Information \& Computation}, vol.~14, no.~3-4, pp.~181--216, 2014.

\bibitem{van2013efficient}
M.~Van Den~Nest, ``Efficient classical simulations of quantum fourier
  transforms and normalizer circuits over abelian groups,'' {\em Quantum
  Information \& Computation}, vol.~13, no.~11-12, pp.~1007--1037, 2013.

\bibitem{cong2017universal}
I.~Cong, M.~Cheng, and Z.~Wang, ``Universal quantum computation with gapped
  boundaries,'' {\em Physical Review Letters}, vol.~119, no.~17, p.~170504,
  2017.

\bibitem{cong2016topological}
I.~Cong, M.~Cheng, and Z.~Wang, ``Topological quantum computation with gapped
  boundaries,'' {\em arXiv preprint arXiv:1609.02037}, 2016.

\bibitem{hutter2015parafermions}
A.~Hutter, J.~R. Wootton, and D.~Loss, ``Parafermions in a kagome lattice of
  qubits for topological quantum computation,'' {\em Physical Review X},
  vol.~5, no.~4, p.~041040, 2015.

\bibitem{hutter2016quantum}
A.~Hutter and D.~Loss, ``Quantum computing with parafermions,'' {\em Physical
  Review B}, vol.~93, no.~12, p.~125105, 2016.

\bibitem{wootton2011engineering}
J.~R. Wootton, V.~Lahtinen, B.~Doucot, and J.~K. Pachos, ``Engineering complex
  topological memories from simple abelian models,'' {\em Annals of Physics},
  vol.~326, no.~9, pp.~2307--2314, 2011.

\bibitem{wootton2010dissecting}
J.~R. Wootton, {\em Dissecting topological quantum computation}.
\newblock University of Leeds, 2010.

\bibitem{wall1963quadratic}
C.~T.~C. Wall, ``Quadratic forms on finite groups, and related topics,'' {\em
  Topology}, vol.~2, no.~4, pp.~281--298, 1963.

\bibitem{basak2015indicators}
T.~Basak and R.~Johnson, ``Indicators of tambara--yamagami categories and gauss
  sums,'' {\em Algebra \& Number Theory}, vol.~9, no.~8, pp.~1793--1823, 2015.

\bibitem{galindo2016solutions}
C.~Galindo and N.~Jaramillo, ``Solutions of the hexagon equation for abelian
  anyons,'' {\em Revista Colombiana de Matem{\'a}ticas}, vol.~50, no.~2,
  pp.~277--298, 2016.

\bibitem{wang2020and}
L.~Wang and Z.~Wang, ``In and around abelian anyon models,'' {\em Journal of
  Physics A: Mathematical and Theoretical}, vol.~53, no.~50, p.~505203, 2020.

\bibitem{nest2012efficient}
M.~Nest, ``Efficient classical simulations of quantum fourier transforms and
  normalizer circuits over abelian groups,'' {\em arXiv preprint
  arXiv:1201.4867}, 2012.

\bibitem{bermejo2012classical}
J.~Bermejo-Vega and M.~V.~d. Nest, ``Classical simulations of abelian-group
  normalizer circuits with intermediate measurements,'' {\em arXiv preprint
  arXiv:1210.3637}, 2012.

\bibitem{horsman2012surface}
C.~Horsman, A.~G. Fowler, S.~Devitt, and R.~Van~Meter, ``Surface code quantum
  computing by lattice surgery,'' {\em New Journal of Physics}, vol.~14,
  no.~12, p.~123011, 2012.

\bibitem{litinski2019magic}
D.~Litinski, ``Magic state distillation: Not as costly as you think,'' {\em
  Quantum}, vol.~3, p.~205, 2019.

\bibitem{wang2010topological}
Z.~Wang, {\em Topological quantum computation}.
\newblock No.~112, American Mathematical Soc., 2010.

\bibitem{hillar2007automorphisms}
C.~J. Hillar and D.~L. Rhea, ``Automorphisms of finite abelian groups,'' {\em
  The American Mathematical Monthly}, vol.~114, no.~10, pp.~917--923, 2007.

\bibitem{moussa2015quantum}
J.~E. Moussa, ``Quantum circuits for qubit fusion,'' {\em Quantum Information
  \& Computation}, vol.~16, no.~SAND-2015-11052J, 2015.

\bibitem{moussa2016transversal}
J.~E. Moussa, ``Transversal clifford gates on folded surface codes,'' {\em
  Physical Review A}, vol.~94, no.~4, p.~042316, 2016.

\bibitem{kesselring2018boundaries}
M.~S. Kesselring, F.~Pastawski, J.~Eisert, and B.~J. Brown, ``The boundaries
  and twist defects of the color code and their applications to topological
  quantum computation,'' {\em Quantum}, vol.~2, p.~101, 2018.

\bibitem{scruby2020hierarchy}
T.~Scruby and D.~E. Browne, ``A hierarchy of anyon models realised by twists in
  stacked surface codes,'' {\em Quantum}, vol.~4, p.~251, 2020.

\bibitem{bombin2023logical}
H.~Bombin, C.~Dawson, R.~V. Mishmash, N.~Nickerson, F.~Pastawski, and
  S.~Roberts, ``Logical blocks for fault-tolerant topological quantum
  computation,'' {\em PRX Quantum}, vol.~4, no.~2, p.~020303, 2023.

\bibitem{laubscher2019universal}
K.~Laubscher, D.~Loss, and J.~R. Wootton, ``Universal quantum computation in
  the surface code using non-abelian islands,'' {\em Physical Review A},
  vol.~100, no.~1, p.~012338, 2019.

\bibitem{paetznick2013universal}
A.~Paetznick and B.~W. Reichardt, ``Universal fault-tolerant quantum
  computation with only transversal gates and error correction,'' {\em Physical
  review letters}, vol.~111, no.~9, p.~090505, 2013.

\bibitem{bombin2015gauge}
H.~Bomb{\'\i}n, ``Gauge color codes: optimal transversal gates and gauge fixing
  in topological stabilizer codes,'' {\em New Journal of Physics}, vol.~17,
  no.~8, p.~083002, 2015.

\bibitem{brown2020fault}
B.~J. Brown, ``A fault-tolerant non-clifford gate for the surface code in two
  dimensions,'' {\em Science advances}, vol.~6, no.~21, p.~eaay4929, 2020.

\bibitem{mochon2004anyon}
C.~Mochon, ``Anyon computers with smaller groups,'' {\em Physical Review A},
  vol.~69, no.~3, p.~032306, 2004.

\bibitem{cui2015universal}
S.~X. Cui, S.-M. Hong, and Z.~Wang, ``Universal quantum computation with weakly
  integral anyons,'' {\em Quantum Information Processing}, vol.~14,
  pp.~2687--2727, 2015.

\bibitem{cui2020kitaev}
S.~X. Cui, D.~Ding, X.~Han, G.~Penington, D.~Ranard, B.~C. Rayhaun, and
  Z.~Shangnan, ``Kitaev's quantum double model as an error correcting code,''
  {\em Quantum}, vol.~4, p.~331, 2020.

\bibitem{iqbal2023creation}
M.~Iqbal, N.~Tantivasadakarn, R.~Verresen, S.~L. Campbell, J.~M. Dreiling,
  C.~Figgatt, J.~P. Gaebler, J.~Johansen, M.~Mills, S.~A. Moses, {\em et~al.},
  ``Creation of non-abelian topological order and anyons on a trapped-ion
  processor,'' {\em arXiv preprint arXiv:2305.03766}, 2023.

\bibitem{beigi2011quantum}
S.~Beigi, P.~W. Shor, and D.~Whalen, ``The quantum double model with boundary:
  condensations and symmetries,'' {\em Communications in mathematical physics},
  vol.~306, pp.~663--694, 2011.

\end{thebibliography}

\end{document}